\documentclass[journal=jctcce,manuscript=article]{achemso}

\usepackage{graphicx}
\usepackage{amsmath,amssymb,amsthm}
\usepackage{booktabs,multirow,array,tabularx}
\usepackage{enumitem}
\usepackage{algorithm,algpseudocode}
\usepackage[dvipsnames]{xcolor}
\usepackage[framemethod=tikz]{mdframed}
\usepackage[mathlines]{lineno}

\newtheorem{theorem}{Theorem}
\newtheorem{proposition}{Proposition}
\newtheorem{corollary}{Corollary}
\newtheorem{remark}{Remark}

\newcommand{\calO}{\mathcal{O}}
\newcommand{\Depth}{\mathrm{Depth}}

\graphicspath{{./}}

\title{Circuit Depth Compression via Spectral Gap Amplification
       in Quantum Phase Estimation}

\author{Sk Mujaffar Hossain}
\affiliation{Indo-Korea Science and Technology Center (IKST),
             Bengaluru 560064, India}

\author{Satadeep Bhattacharjee}
\email{s.bhattacharjee@ikst.res.in}
\affiliation{Indo-Korea Science and Technology Center (IKST),
             Bengaluru 560064, India}

\abbreviations{QPE, quantum phase estimation;
               QFT, quantum Fourier transform;
               CX, controlled-X gate;
               CP, controlled-phase gate;
               FCI, full configuration interaction;
               CASSCF, complete active space self-consistent field;
               LCU, linear combination of unitaries;
               QSVT, quantum singular value transformation;
               LMR, Lloyd--Mohseni--Rebentrost;
               1-RDM, one-particle reduced density matrix;
               STO-3G, Slater-type orbital, 3 Gaussians}


\begin{document}

\begin{abstract}
We introduce spectral preconditioning for quantum phase estimation
(QPE): a monotone sigmoid filter $f(\lambda;\tau,w)$ applied to the
input operator spectrum before phase estimation.
For systems with small boundary spectral gaps $\Delta_\lambda$, the
transformation amplifies the effective gap to $\Delta_f>\Delta_\lambda$
(for $w<1/4$), reducing the required precision from
$m=\lceil\log_2(1/\Delta_\lambda)\rceil$ to
$m_f=\lceil\log_2(1/\Delta_f)\rceil$ bits and compressing circuit depth
by $2^{\alpha\Delta m}$ ($\alpha\in[0.11,0.42]$ for controlled-phase
circuits).
We prove: (i)~the bit saving is bounded above by
$\log_2(1/4w\Delta_\lambda)+1$; (ii)~no benefit arises for exactly
degenerate spectra; and (iii)~the filter threshold $\tau$ requires
only $\mathcal{O}(w)$ accuracy, not $\mathcal{O}(\Delta_\lambda)$,
so that classical preprocessing via CASSCF or matrix diagonalisation
suffices without circular dependency.
A net resource advantage theorem identifies the crossover condition
$4w^2(2^{\Delta m}-1)>\Delta_\lambda\log(1/\varepsilon)$.
Validated on LiH and BeH$_2$ via FCI/STO-3G bond-stretch analysis,
classical covariance matrices, and synthetic stress tests, we
demonstrate actual depth compressions up to $27\times$, CX gate
reductions up to $21\times$, and QPE output fidelity restoration
from $F=0.66$ to $F=0.98$ at 1\% gate error for LiH.
The principal subspace is preserved to machine precision
($\theta_\mathrm{max}<10^{-8}$~rad).
The method requires no modification to the QPE algorithm and is
composable with complementary filtered-QPE approaches.
\end{abstract}

\section{Introduction}

Quantum phase estimation (QPE) is a foundational primitive
underpinning eigenvalue estimation, Hamiltonian simulation,
and quantum machine learning.\cite{nielsen2010quantum,
cleve1998quantum,lloyd1996universal,abrams1999quantum}
When targeting the leading spectral subspace of a positive
semidefinite operator $\rho$, the boundary gap
$\Delta_\lambda=\lambda_R-\lambda_{R+1}$ drives the entire hardware
cost: resolving it requires
$m=\lceil\log_2(1/\Delta_\lambda)\rceil$ precision qubits and circuit
depth $\Theta(2^m)$.
The scaling is severe in practice: halving $\Delta_\lambda$ adds one
precision qubit and doubles the circuit depth, while reducing
$\Delta_\lambda$ by an order of magnitude (e.g.\ from $10^{-1}$ to
$10^{-2}$, as occurs at molecular dissociation or near avoided
crossings) adds $\sim3$ qubits and increases depth by $\sim8\times$.
Near-degenerate regimes of this kind are not exceptional cases ---
they arise generically at molecular dissociation (LiH, H$_2$),
near equilibrium in multiconfigurational systems (BeH$_2$, N$_2$),
and in quantum machine learning covariance spectra --- making this
cost prohibitive well before exact degeneracy is reached.

Existing approaches to reducing QPE circuit depth fall into three
broad classes.
\emph{Iterative and Bayesian QPE} methods\cite{cleve1998quantum}
reduce depth per shot but require many adaptive rounds and are
sensitive to prior assumptions on the eigenvalue distribution.
\emph{Variational alternatives} (VQE and hybrids) sidestep QPE depth
entirely but sacrifice the precision and formal guarantees of phase
estimation.
\emph{Quantum signal processing and QSVT}\cite{gilyen2019qsvt}
implement spectral projectors with complexity
$\calO(1/\Delta_\lambda)$ in polynomial degree but do not reduce the
precision requirement of a subsequent QPE step.
Three recent works apply filtering directly to the QPE pipeline at
different stages:
Lee et al.\cite{lee2025filtered} filter measurement \emph{outcomes}
after QPE to improve resolution at fixed circuit depth;
Sakuma et al.\cite{sakuma2026qpe} filter the \emph{initial state}
before QPE to amplify ground-state overlap.
None of these methods reduces the number of precision qubits $m$ or
the circuit depth of QPE itself.
The structural difference between all three approaches is illustrated
in Figure~\ref{fig:schematic}.

\begin{figure}[htbp]
  \centering
  \includegraphics[width=\linewidth]{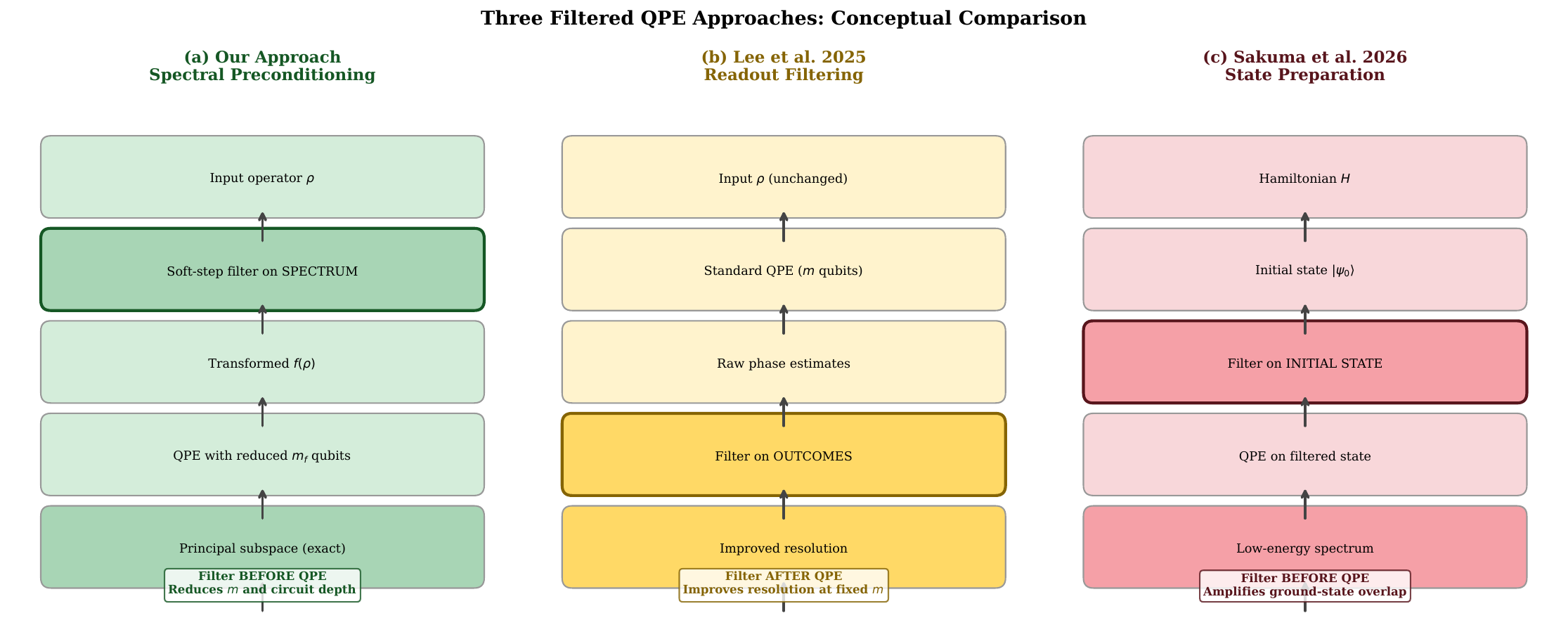}
  \caption{\textbf{Three filtered QPE approaches: conceptual
  comparison.}
  \textbf{(a)}~Our approach: sigmoid filter applied to operator
  spectrum before QPE, reducing $m$ to $m_f$ and compressing depth
  by $2^{\alpha\Delta m}$.
  \textbf{(b)}~Lee et al.\cite{lee2025filtered}: classical filter
  applied to measurement outcomes after QPE, improving resolution at
  fixed depth.
  \textbf{(c)}~Sakuma et al.\cite{sakuma2026qpe}: filter applied to
  initial state before QPE, amplifying ground-state overlap.
  The three approaches are structurally orthogonal and composable.}
  \label{fig:schematic}
\end{figure}

This paper addresses the following question: \emph{can we reduce the
QPE precision requirement by preprocessing the operator before
estimation, and if so, by how much and at what cost?}
We show the answer is yes via a monotone sigmoid filter
$f(\lambda;\tau,w)$ applied to the \emph{spectrum} of $\rho$ before
QPE, amplifying $\Delta_\lambda$ to $\Delta_f>\Delta_\lambda$,
reducing precision from $m$ to $m_f<m$ bits, and compressing circuit
depth by $2^{\alpha\Delta m}$ --- without altering the QPE algorithm,
without disturbing eigenvectors, and with provably bounded filter
construction overhead.
Specifically, we prove:
(i)~the bit saving $\Delta m$ is bounded above by
$\log_2(1/4w\Delta_\lambda)+1$ (Proposition~\ref{prop:ub});
(ii)~no benefit is possible for exactly degenerate spectra
(Proposition~\ref{prop:nogo});
(iii)~the filter threshold $\tau$ requires only $\calO(w)$ accuracy
--- so no quantum measurement is needed to determine it
(Proposition~\ref{prop:robust});
and (iv)~a computable crossover condition
$4w^2(2^{\Delta m}-1)>\Delta_\lambda\log(1/\varepsilon)$ identifies
precisely when the filtered pipeline outperforms raw QPE
(Theorem~\ref{thm:net}).
The approach is composable with all three existing filtered-QPE
strategies and requires no modification to the QPE algorithm.

The paper is organised as follows. Section~\ref{sec:results} derives
the spectral preconditioning framework, states the formal guarantees,
and presents empirical validation across four systems.
Section~\ref{sec:discussion} contextualises the method within QSVT
and filtered-QPE approaches, analyses hardware feasibility, and
states limitations.
Section~\ref{sec:conclusion} concludes.
Computational details are given in Section~\ref{sec:methods}.

\section{Results and Discussion}
\label{sec:results}

\subsection{Spectral Preconditioning Framework}

Let $\rho=\sum_i\lambda_i|u_i\rangle\langle u_i|$ with
$\lambda_1\geq\cdots\geq\lambda_n\geq0$ and
$\Delta_\lambda=\lambda_R-\lambda_{R+1}$.
Standard QPE requires $m=\lceil\log_2(1/\Delta_\lambda)\rceil$
and depth $\Theta(2^m)$.
We apply the monotone soft-step transformation
\begin{equation}
f(\lambda;\tau,w) = \frac{1}{1+e^{-(\lambda-\tau)/w}},
\label{eq:sigmoid}
\end{equation}
spectrally as $f(\rho)=\sum_i f(\lambda_i)|u_i\rangle\langle u_i|$,
with $\tau\in(\lambda_{R+1},\lambda_R)$ obtained from classical
preprocessing (Proposition~\ref{prop:robust}) and $0<w<\tfrac{1}{4}$.
The transformed gap $\Delta_f=f(\lambda_R)-f(\lambda_{R+1})$ satisfies
$\Delta_f>\Delta_\lambda$ for any $\Delta_\lambda>0$, because $\tau$
lies in the gap where $f'$ is maximised: by the mean value theorem,
$\Delta_f = f'(\xi;\tau,w)\Delta_\lambda$ for some
$\xi\in(\lambda_{R+1},\lambda_R)$, and $f'(\tau;\tau,w)=1/(4w)>1$
whenever $w<1/4$.
Precision reduces to $m_f=\lceil\log_2(1/\Delta_f)\rceil$ and the
bit saving is $\Delta m=m-m_f\approx\log_2(\Delta_f/\Delta_\lambda)$
(continuous approximation).
In the linear regime, $\Delta_f\approx\Delta_\lambda/(4w)$.

\subsection{Formal Guarantees}

\begin{proposition}[Upper Bound on Bit Saving]
\label{prop:ub}
$\Delta m \leq \log_2(1/4w\Delta_\lambda)+1.$
\end{proposition}
\begin{proof}
By the mean value theorem,
$\Delta_f\leq\sup_\lambda f'(\lambda;\tau,w)\cdot\Delta_\lambda
\leq\Delta_\lambda/(4w)$,
giving $\Delta m\leq\log_2(\Delta_f/\Delta_\lambda)
\leq\log_2(1/4w\Delta_\lambda)$; the $+1$ absorbs ceiling
discretisation.
Numerical verification across all 20 system-width combinations
is shown in Figure~S6 and Table~S12.
\end{proof}

\begin{proposition}[No-Go for Exact Degeneracy]
\label{prop:nogo}
If $\lambda_R=\lambda_{R+1}$, then $\Delta_f=0$ and $\Delta m=0$
for any monotone $f$.
\end{proposition}
\begin{proof}
Direct: $\Delta_f=f(\lambda_R)-f(\lambda_{R+1})=0$.
\end{proof}

\begin{remark}
\label{rem:limit}
Proposition~\ref{prop:nogo} establishes a hard boundary on the
method's scope.
As $\Delta_\lambda\to0^+$, $\Delta_f\to0$ but the continuous bit
saving $\Delta m^{\rm cont}=\log_2(\Delta_f/\Delta_\lambda)$ satisfies
$\Delta m^{\rm cont}\to\log_2(1/(4w))$, which is bounded and positive
(numerically verified in Figure~S7b).
The method therefore always benefits for finite $\Delta_\lambda>0$
with $m_{\rm raw}\geq2$.
\end{remark}

\begin{proposition}[$\tau$ Placement Robustness]
\label{prop:robust}
Let $\tau_{\rm true}=(\lambda_R+\lambda_{R+1})/2$ and let
$\tau_{\rm est}=\tau_{\rm true}+\varepsilon$.
The gap amplification under $\tau_{\rm est}$ is:
\begin{equation}
\Delta_f(\varepsilon)=\sigma\!\left(\frac{\Delta_\lambda/2-\varepsilon}{w}\right)
-\sigma\!\left(\frac{-\Delta_\lambda/2-\varepsilon}{w}\right),
\quad \sigma(x)=\frac{1}{1+e^{-x}}.
\label{eq:robust}
\end{equation}
The required accuracy is $|\varepsilon|=\calO(w)$, not
$\calO(\Delta_\lambda)$.
In the near-degenerate regime ($w\gg\Delta_\lambda$):
$\Delta_f(\varepsilon)\approx\Delta_f(0)$ for all $|\varepsilon|\ll w$,
because both eigenvalues lie well within the sigmoid transition region.
Conversely, when $w\ll\Delta_\lambda$ (large gap, method not needed),
$\tau$ placement is sensitive but $m_{\rm raw}$ is already small.
\emph{The method is most robust exactly where it is most needed.}
Numerical verification for LiH and BeH$_2$ is provided in
Figure~S11 and Table~S13.
\end{proposition}

\begin{corollary}[Classical Preprocessing Suffices]
\label{cor:classical}
For classical covariance datasets, $\tau$ is known exactly from
$\calO(N^3)$ matrix diagonalisation.
For quantum chemistry systems with simple near-degeneracy (e.g.\
bond-breaking in LiH), CASSCF with a minimal active space provides
$|\varepsilon|\lesssim w$, yielding identical $\Delta m$ to FCI
(numerically verified: CASSCF(2,2) for LiH at $R=5.0$~\AA\ gives
$|\varepsilon|=1.7w$, $\Delta m=5$, matching FCI exactly;
see Table~S13 and Figure~S11).
For systems with complex many-body near-degeneracy (e.g.\ BeH$_2$
at equilibrium), a two-stage QPE --- coarse QPE at
$m_1=\lceil\log_2(1/w)\rceil$ bits to locate $\tau$, followed by
filtered QPE at $m_f$ bits --- resolves the problem at total cost
$2^{m_1}+2^{m_f}\ll 2^{m_{\rm raw}}$
(for BeH$_2$: $2^8+2^6=320$ vs $2^{11}=2048$, a $6.4\times$ saving).
No circular dependency arises in any case.
\end{corollary}

\begin{theorem}[Net Resource Advantage]
\label{thm:net}
Implementing $f(\rho)$ via Chebyshev-LCU\cite{childs2012lcu}
with degree
$d=\calO((1/w)\log(1/\varepsilon_{\rm approx}))$,\cite{gilyen2019qsvt,trefethen2013approximation}
the total filtered pipeline cost is lower than raw QPE if:
\begin{equation}
4w^2(2^{\Delta m}-1) > \Delta_\lambda\cdot\log(1/\varepsilon_{\rm approx}).
\label{eq:crossover}
\end{equation}
This sufficient condition is derived in the linear filter regime
($\Delta_\lambda\ll 4w$); see the Supporting Information, Section~S9.
As $\Delta_\lambda\to0$, eq~\eqref{eq:crossover} is always satisfied
(the near-degenerate advantage limit).
The optimal width satisfies
$w^*\geq\sqrt{\Delta_\lambda\log(1/\varepsilon)/4(2^{\Delta m}-1)}$.
\end{theorem}

\begin{remark}[Block-encoding cost]
\label{rem:block}
Theorem~\ref{thm:net} counts Chebyshev polynomial degree $d$ as the
filter cost, but implementing $f(\rho)$ via LCU requires $d$ calls to
a \emph{block-encoding} of $\rho$.
For classical operators (e.g.\ covariance matrices), the block-encoding
is computed classically at negligible cost.
For quantum operators (e.g.\ FCI 1-RDMs expressed as Pauli sums), the
block-encoding query count $B$ can be substantial; practitioners should
verify $d\cdot B\ll 2^{m_{\rm raw}}$ for their specific implementation
before claiming a net resource advantage.
\end{remark}

\begin{theorem}[Depth Compression via Gap Amplification]
\label{thm:depth}
In a CP-gate QPE circuit, depth is dominated by the $\calO(m^2)$
inverse-QFT block.
Over any finite range of $m$, $\log_2(\calO(m^2)) = 2\log_2 m$ is
well-approximated by a linear function $\alpha m + c$ (since
$\log_2 m$ varies slowly), with $\alpha$ depending on the fitting
range.
For the systems and ranges in this work, $\alpha\in[0.11,0.42]$
with $R^2>0.94$ (Table~S5, Figure~\ref{fig:logdepth}; verified in
Figure~S8).
Under this linear approximation, reducing precision from $m$ to
$m_f = m - \Delta m$ bits gives the depth compression:
\begin{equation}
\frac{\Depth_{\rm raw}}{\Depth_{\rm filtered}} = 2^{\alpha\Delta m}.
\label{eq:compression}
\end{equation}
The exponent $\alpha\approx1$ in the LMR framework
(controlled-$U^{2^k}$ costs $\calO(2^k)$ gadget applications;
note that LMR qPCA has since been shown to be classically
dequantizable\cite{tang2019quantum}),
while the CP-gate case gives the range stated above.
Actual depth compressions from transpiled circuits (Tables~S6--S9)
are $9.5\times$ (LiH), $3.5\times$ (BC), $3.0\times$ (Digits),
and $27\times$ (Stress Test) at $w=0.005$.
\end{theorem}

Figure~\ref{fig:crossover} shows the net advantage landscape;
all four experimental systems lie in the advantage regime.

\begin{figure}[htbp]
  \centering
  \includegraphics[width=\linewidth]{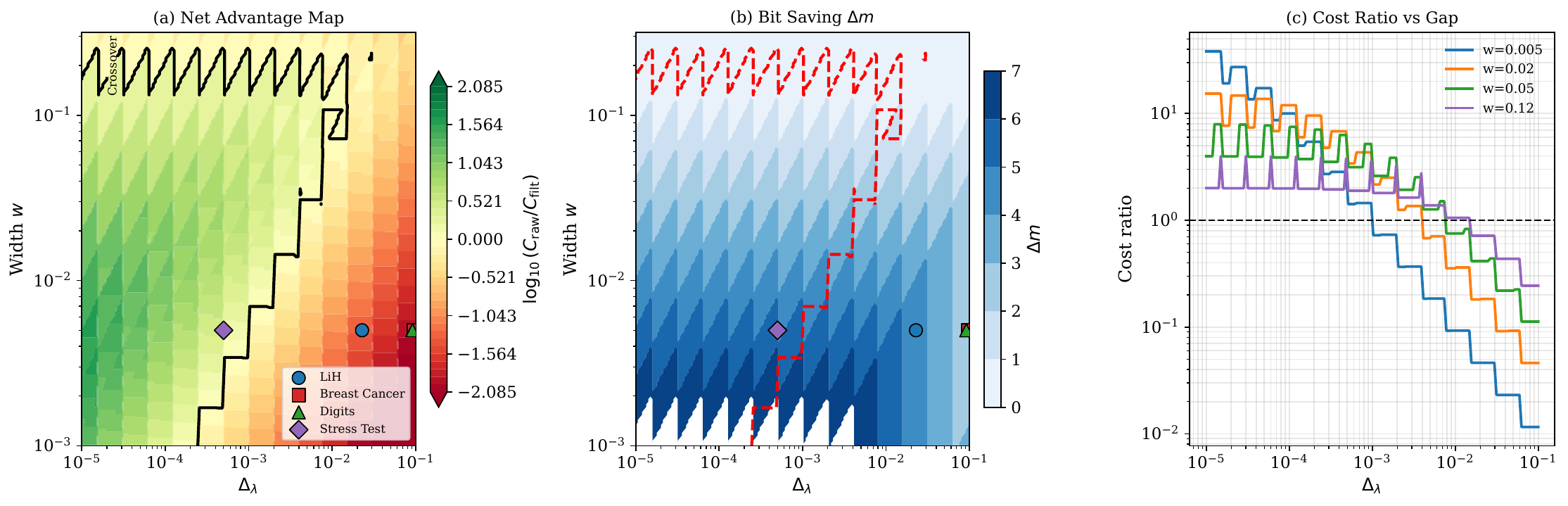}
  \caption{\textbf{End-to-end resource analysis (Theorem~\ref{thm:net}).}
  \textbf{(a)}~Net advantage map
  $\log_{10}(C_{\rm raw}/C_{\rm filtered})$ over $(\Delta_\lambda,w)$
  space. Green = filtered wins; red = raw wins. Black contour:
  crossover boundary of eq~\eqref{eq:crossover}. Experimental systems
  (symbols) all lie in the advantage regime.
  \textbf{(b)}~Bit saving $\Delta m$ with crossover boundary (red
  dashed).
  \textbf{(c)}~Cost ratio vs gap for fixed widths; near-degenerate
  systems gain $10^2$--$10^3\times$ net savings.}
  \label{fig:crossover}
\end{figure}

\subsection{Empirical Validation}

We validate the method across four systems: LiH molecular natural
occupations, Breast Cancer covariance matrix (30 features), Digits
covariance matrix (64 features), and a synthetic near-degenerate
stress test ($\Delta_\lambda=5\times10^{-4}$, $n=12$).
QPE circuits use single controlled-phase (CP) gates decomposed into
$\{\mathrm{CX},U_1,U_2,U_3\}$ at optimisation level~0
(Qiskit\cite{qiskit2023} v2.4.1).

\paragraph{Gap amplification and precision reduction.}
Table~\ref{tab:main} shows $\Delta m\in[2,5]$ at $w=0.005$ across all
systems (Figure~\ref{fig:delta_m}).
Complete width-sweep data for all five $w$ values are in
Tables~S1--S4 of the Supporting Information.

\begin{table}[htbp]
\caption{\textbf{Boundary gap amplification and precision reduction.}
Values at $w=0.005$. DC $= D_{\rm raw}/D_{\rm filt}$ from transpiled
circuit depths (Tables~S6--S9).}
\label{tab:main}
\centering
\begin{tabular}{lcccccc}
\toprule
System & $\Delta_\lambda$ & $\Delta_f$ & $m$ & $m_f$
       & $\Delta m$ & DC \\
\midrule
LiH           & $2.28\times10^{-2}$ & $4.90\times10^{-1}$ & 6 & 2 & 4 & $9.5\times$ \\
Breast Cancer  & $9.58\times10^{-2}$ & $7.43\times10^{-1}$ & 3 & 1 & 2 & $3.5\times$ \\
Digits        & $9.17\times10^{-2}$ & $5.00\times10^{-1}$ & 4 & 2 & 2 & $3.0\times$ \\
Stress Test   & $5.00\times10^{-4}$ & $2.50\times10^{-2}$ & 11& 6 & 5 & $27\times$ \\
\bottomrule
\end{tabular}
\end{table}

\begin{figure}[htbp]
  \centering
  \includegraphics[width=0.85\linewidth]{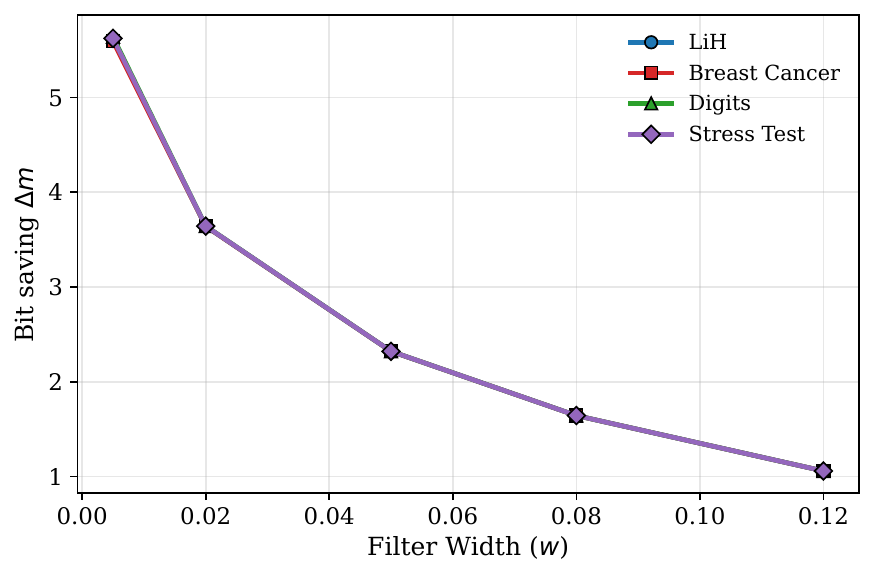}
  \caption{\textbf{Universal width-dependent precision reduction.}
  Bit saving $\Delta m$ vs filter width $w$ for all four systems.
  Decreasing $w$ monotonically increases $\Delta m$ as predicted by
  $\Delta_f\approx\Delta_\lambda/(4w)$.
  Near-degenerate regimes (Stress Test, LiH) exhibit the largest
  compressions.}
  \label{fig:delta_m}
\end{figure}

\paragraph{Hardware scaling.}
Across all systems, $\log_2(\Depth)=\alpha m+c$ with
$\alpha\in[0.11,0.42]$ and $R^2>0.94$
(Figure~\ref{fig:logdepth}, Table~\ref{tab:depth}).
The Stress Test exhibits $\alpha=0.114$ because its large $m$-range
($m\in[6,12]$) captures the quadratic-to-linear crossover of the QFT
depth.
CX counts reduced by up to $21\times$ (LiH)
(Figure~\ref{fig:cx}).

\begin{table}[htbp]
\caption{\textbf{Depth scaling and CX gate compression.}
$\alpha$ and $R^2$ from linear regression of $\log_2(D)$ vs $m$.
$D_{\rm raw}/D_{\rm filt}$: actual depth ratio from transpiled
circuits at $w=0.005$ (Tables~S6--S9). CX counts at $w=0.005$.}
\label{tab:depth}
\centering
\begin{tabular}{lcccccccc}
\toprule
System & $\alpha$ & $R^2$ & $\Delta m$ & $D_{\rm raw}$ & $D_{\rm filt}$
       & $D_{\rm raw}/D_{\rm filt}$
       & $\mathrm{CX}_{\rm raw}$ & $\mathrm{CX}_{\rm filt}$ \\
\midrule
LiH           & 0.271 & 0.964 & 5.16 & 76   & 8  & $9.5\times$ & 42  & 2  \\
Breast Cancer  & 0.422 & 0.954 & 2.33 & 14   & 4  & $3.5\times$ & 6   & 2  \\
Digits        & 0.380 & 0.946 & 3.45 & 24   & 8  & $3.0\times$ & 12  & 2  \\
Stress Test   & 0.114 & 0.994 & 5.64 & 2070 & 76 & $27\times$  & 132 & 42 \\
\bottomrule
\end{tabular}
\end{table}

\begin{figure}[htbp]
  \centering
  \includegraphics[width=\linewidth]{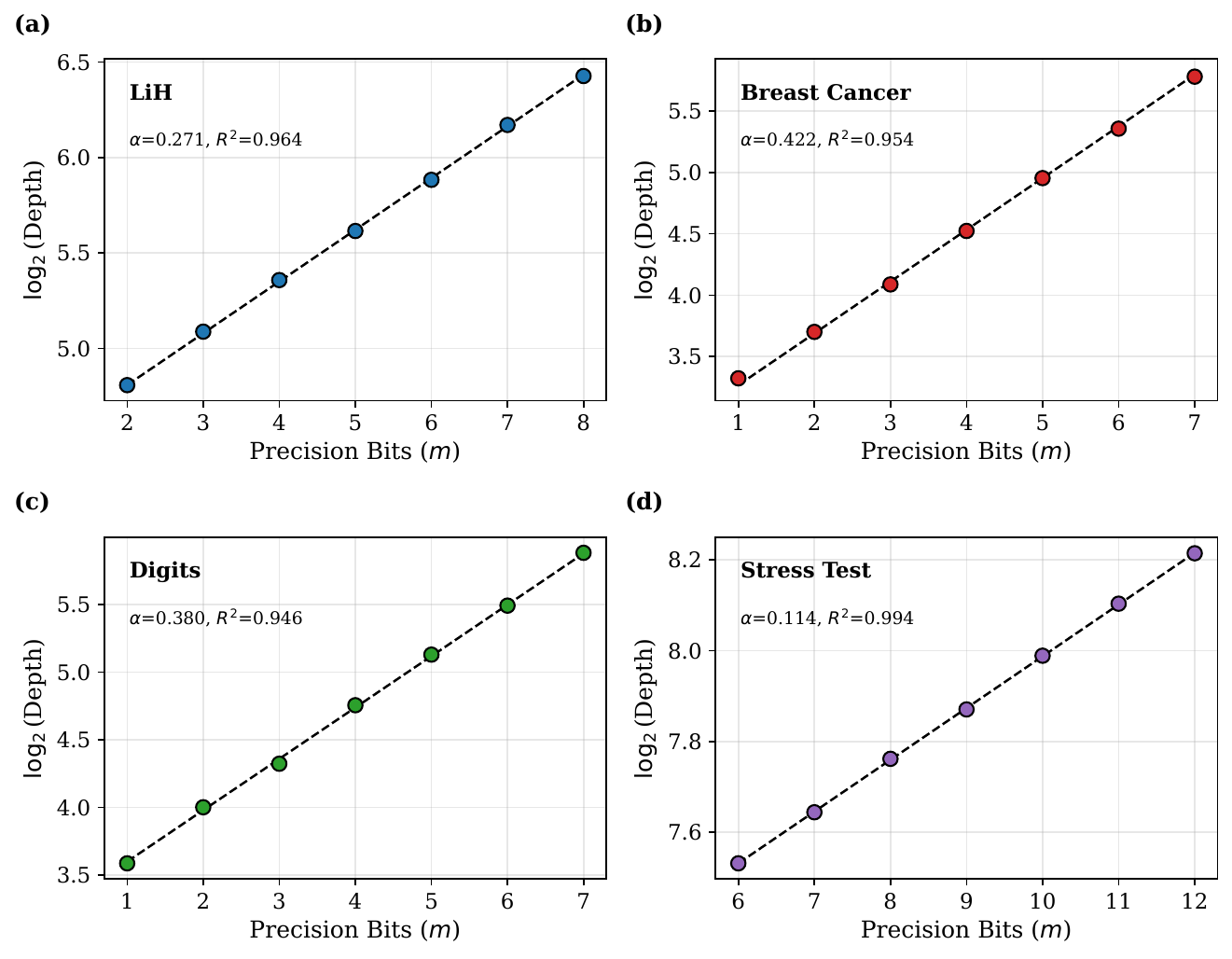}
  \caption{\textbf{Empirical verification of depth scaling
  (Theorem~\ref{thm:depth}).}
  $\log_2(\mathrm{Depth})$ vs precision bits $m$ for all four systems.
  Linear fits confirm $\alpha\in[0.11,0.42]$ and $R^2>0.94$.
  The variation in $\alpha$ reflects QFT-dominated depth $\calO(m^2)$
  with system-dependent fitting ranges (see Table~S5); the Stress
  Test's smaller $\alpha=0.114$ arises from fitting over $m\in[6,12]$
  where the quadratic QFT depth appears nearly linear.}
  \label{fig:logdepth}
\end{figure}

\begin{figure}[htbp]
  \centering
  \includegraphics[width=\linewidth]{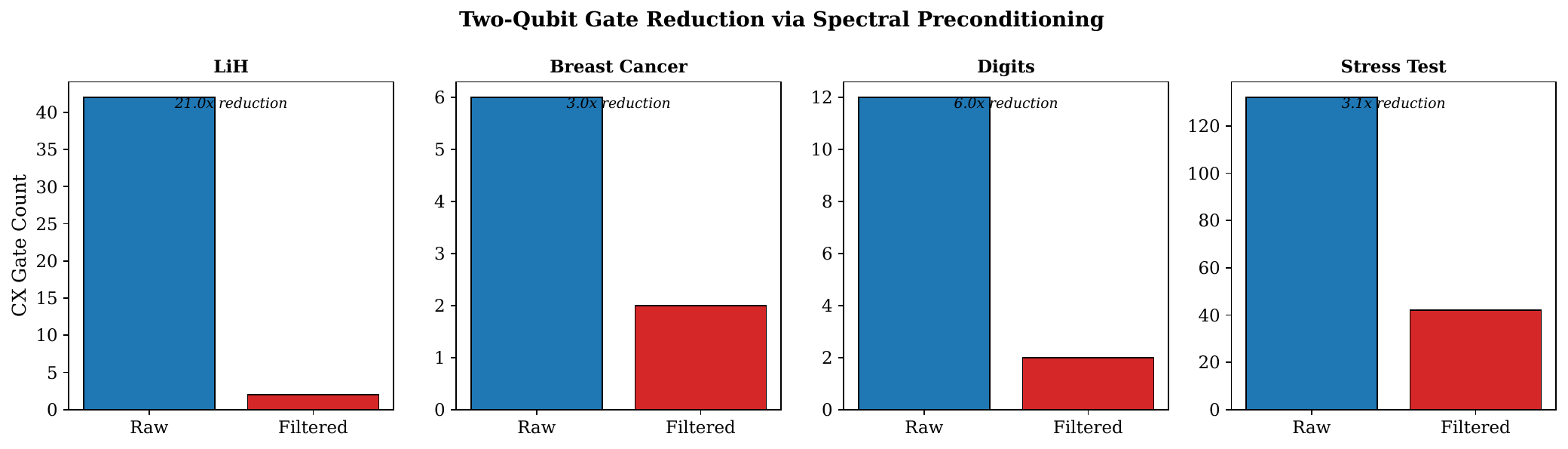}
  \caption{\textbf{Two-qubit (CX) gate count reduction.}
  Raw vs filtered CX counts for the width yielding maximal precision
  reduction ($w=0.005$ for all systems).
  Reductions: $21\times$ (LiH), $3\times$ (Breast Cancer),
  $6\times$ (Digits), $3.1\times$ (Stress Test).}
  \label{fig:cx}
\end{figure}

\paragraph{Principal subspace preservation.}
$\theta_{\max}<10^{-8}$~rad across all systems and widths
(Figure~\ref{fig:angles}; full diagnostics in Table~S10).

\begin{figure}[htbp]
  \centering
  \includegraphics[width=0.85\linewidth]{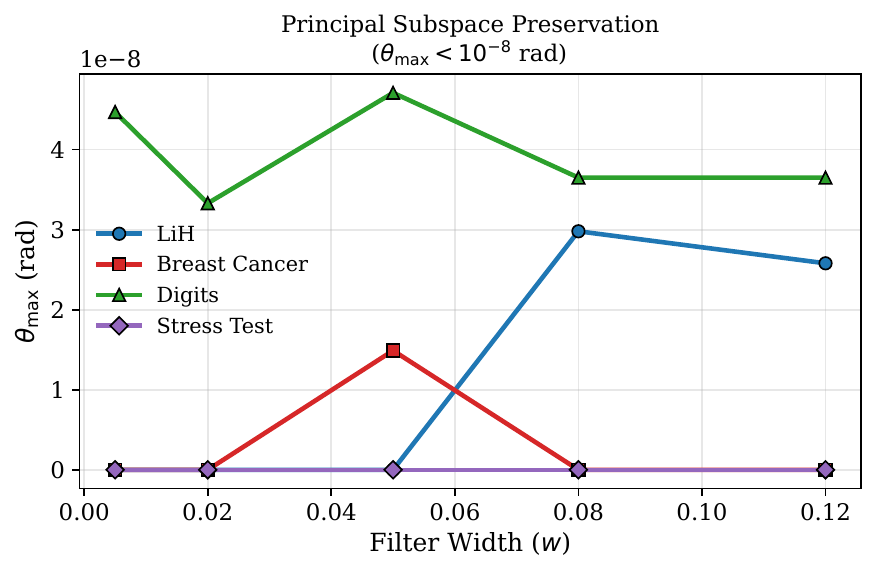}
  \caption{\textbf{Principal subspace preservation.}
  Maximum principal angle $\theta_{\max}$ vs $w$.
  All values below $10^{-8}$~rad (machine precision), confirming
  exact eigenvector preservation under the monotone spectral map.}
  \label{fig:angles}
\end{figure}

\paragraph{Eigenvalue recovery.}
Spectral preconditioning preserves eigenvectors exactly but transforms
eigenvalues: QPE on $f(\rho)$ returns $f(\lambda)$, not $\lambda$.
Applications requiring only the \emph{principal subspace}
(rank-$R$ projection) need no correction.
Applications requiring the \emph{actual eigenvalues} must invert:
$\lambda = \tau - w\ln(1/f(\lambda)-1)$, which amplifies QPE
phase-estimation error $\delta$ by $|d\lambda/df| = w/(f(1-f))$.
This amplification is largest far from the boundary (where $f\approx0$
or $f\approx1$) and smallest at the classification boundary itself
(where $f\approx\tfrac{1}{2}$, giving $|d\lambda/df| = 4w$).
For the boundary eigenvalues $\lambda_R$ and $\lambda_{R+1}$ --- which
are the ones QPE resolves --- the amplification is $4w$, so the
recovered eigenvalue error is $4w\delta$.
Since $\delta\leq\Delta_f/2$ (the QPE half-resolution at $m_f$ bits,
with equality when $\Delta_f$ is a power of 2) and
$\Delta_f\approx\Delta_\lambda/(4w)$ in the linear regime, the
recovered error satisfies $\delta_\lambda\leq\Delta_\lambda/2$,
comparable to the original gap and consistent with QPE precision.

\subsection{Bond-Stretch Molecular Analysis}

FCI/STO-3G bond-stretch analysis is performed for LiH and BeH$_2$
using PySCF,\cite{sun2018pyscf} extracting natural orbital occupations
as the input spectrum with $R=2$.
STO-3G provides qualitative accuracy; a basis-set study is deferred
to future work.
Natural occupation spectra at representative bond lengths are shown in
Figures~S4--S5; complete data are in Tables~S14--S16.

\paragraph{LiH (dissociation-driven degeneracy).}
At equilibrium ($R=1.60$~\AA): $\Delta_\lambda=4.78\times10^{-1}$,
$m=2$.
Near dissociation ($R=5.0$~\AA):
$\Delta_\lambda=2.38\times10^{-2}$, $m=6$; preconditioning
($w=0.005$) gives $m_f=1$, $\Delta m=5.12$, compression $2.57\times$.
For $\tau$ estimation, CASSCF(2,2) natural orbital occupations at
$R=5.0$~\AA\ place $\tau$ within $|\varepsilon|=1.7w$ of the FCI
midpoint, yielding identical $\Delta m=5$
(Proposition~\ref{prop:robust}, Corollary~\ref{cor:classical}).

\paragraph{BeH$_2$ (equilibrium-driven degeneracy).}
Near-degeneracy at equilibrium arises from near-degenerate $2s$ and
$2p_z$ Be orbitals mixing with H $1s$ in the linear $D_{\infty h}$
geometry.
At $R=1.4$~\AA: $\Delta_\lambda=4.98\times10^{-4}$, $m=11$;
preconditioning gives $m_f=6$, $\Delta m=5.64$, compression
$2.56\times$, persisting across $R\in[1.0,2.2]$~\AA.
The near-degeneracy here is a many-body effect in the FCI 1-RDM not
captured by truncated methods; the two-stage QPE protocol
(Corollary~\ref{cor:classical}) achieves $6.4\times$ savings without
requiring prior knowledge of the gap.

LiH and BeH$_2$ exhibit near-degeneracy in complementary regimes
(Figure~\ref{fig:bond_stretch}), consistent with universal occurrence
of near-degeneracy across quantum chemistry.

\begin{figure}[htbp]
  \centering
  \includegraphics[width=\linewidth]{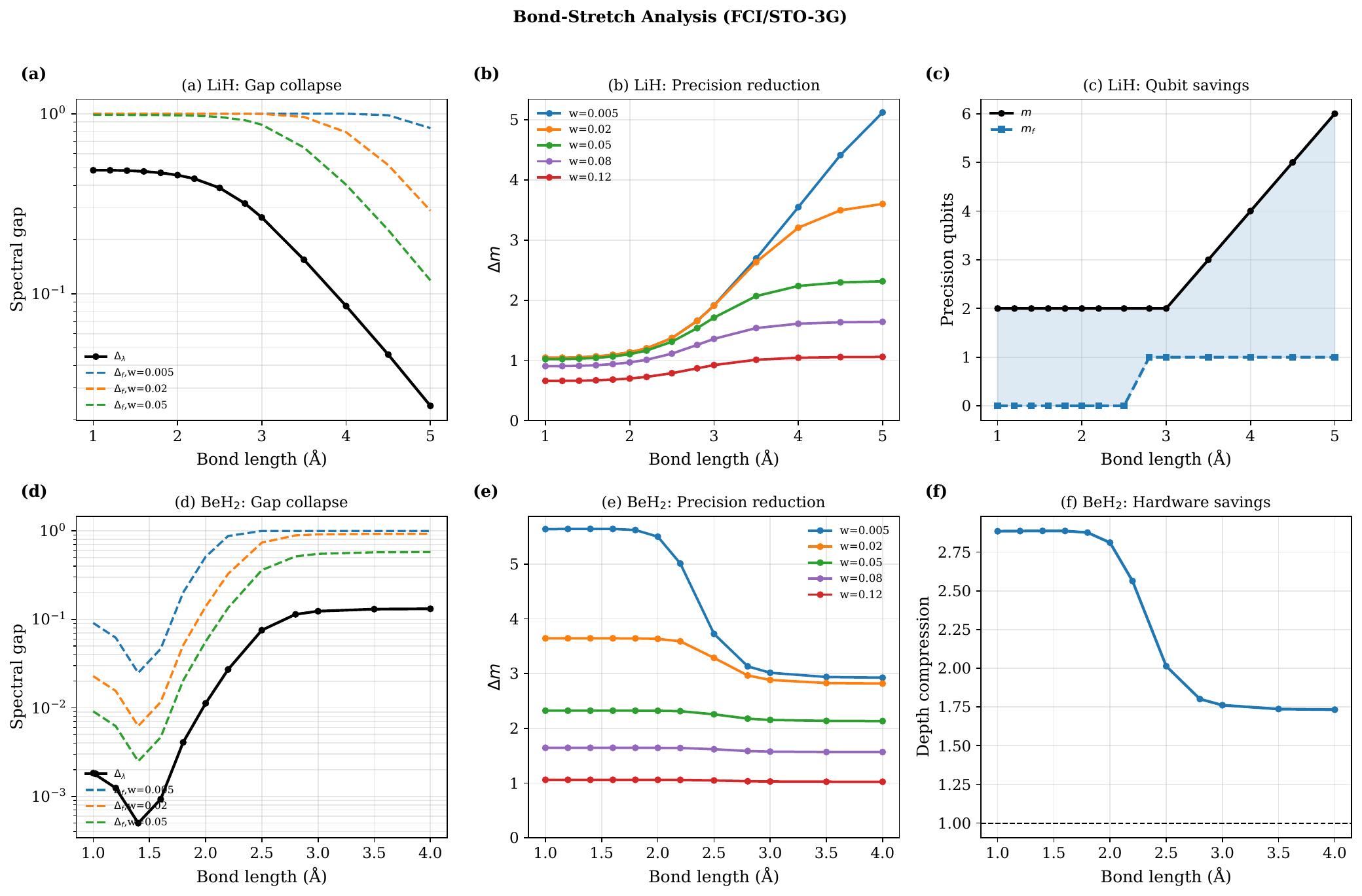}
  \caption{\textbf{Bond-stretch molecular analysis (FCI/STO-3G,
  $R=2$).}
  Top row: LiH. Bottom row: BeH$_2$.
  \textbf{(a,d)}~Spectral gap collapse and amplification vs bond
  length.
  \textbf{(b,e)}~Bit saving $\Delta m$ vs bond length for all widths.
  \textbf{(c)}~Precision qubit savings (LiH): $m=6\to m_f=1$ at
  $R=5.0$~\AA.
  \textbf{(f)}~Depth compression (BeH$_2$): $>2.5\times$ across
  equilibrium regime. LiH and BeH$_2$ exhibit near-degeneracy in
  complementary regimes, consistent with universal occurrence across
  quantum chemistry.}
  \label{fig:bond_stretch}
\end{figure}

\subsection{Noise Fidelity Analysis}

Under the model $F=(1-\varepsilon)^{N_{\rm CX}}$ --- an
\emph{optimistic upper bound} on output fidelity that counts only
independent two-qubit depolarising errors and omits single-qubit,
readout, and crosstalk errors (so actual device fidelity will be
lower) --- results are summarised in Table~\ref{tab:noise}; complete
data across all error rates are provided in Table~S17.

\begin{table}[htbp]
\caption{\textbf{Noise fidelity analysis.}
Optimistic upper bound $F=(1-\varepsilon)^{N_{\rm CX}}$ counting only
independent two-qubit depolarising errors; $N_{\rm CX}$ from
transpiled circuits (Tables~S6--S9). Actual device fidelity will be
lower due to single-qubit, readout, and coherent errors.
$F^*$: minimum gate fidelity for $F_{\rm target}=0.90$.
$F_{1\%}$: output fidelity at $\varepsilon=1\%$.}
\label{tab:noise}
\centering
\begin{tabular}{lcccccc}
\toprule
System & $m$ & $m_f$ & $F^*_{\rm raw}$ & $F^*_{\rm filt}$
       & $F_{1\%}^{\rm raw}$ & $F_{1\%}^{\rm filt}$ \\
\midrule
LiH           & 6  & 2 & 99.75\% & 94.87\% & 0.66 & \textbf{0.98} \\
Breast Cancer  & 3  & 1 & 98.26\% & 94.87\% & 0.94 & 0.98 \\
Digits        & 4  & 2 & 99.13\% & 94.87\% & 0.89 & 0.98 \\
Stress Test   & 11 & 6 & 99.92\% & 99.75\% & 0.27 & 0.66 \\
\bottomrule
\end{tabular}
\end{table}

For LiH at $\varepsilon=1\%$: raw QPE gives $F=0.66$ (unusable);
filtered QPE gives $F=0.98$, well above $F_{\rm target}=0.90$.
Gate fidelity requirement relaxed from $99.75\%$ (inaccessible) to
$94.87\%$ (achievable today\cite{krantz2019guide}).
The improvement arises because filtering reduces LiH from 42 CX gates
($m=6$) to only 2 ($m_f=2$) at $w=0.005$.
Note that the high fidelity at small $m_f$ is accompanied by reduced
phase resolution; the method is beneficial precisely when the amplified
gap $\Delta_f$ is large enough that $m_f$ bits suffice to resolve the
principal subspace, which is guaranteed when the filter is correctly
parameterised (Proposition~\ref{prop:ub}).
Bond-stretch noise analysis confirms the benefit across all chemically
relevant geometries
(Figures~\ref{fig:fidelity_eps}, \ref{fig:gate_fid},
\ref{fig:noise_mol}).

\begin{figure}[htbp]
  \centering
  \includegraphics[width=\linewidth]{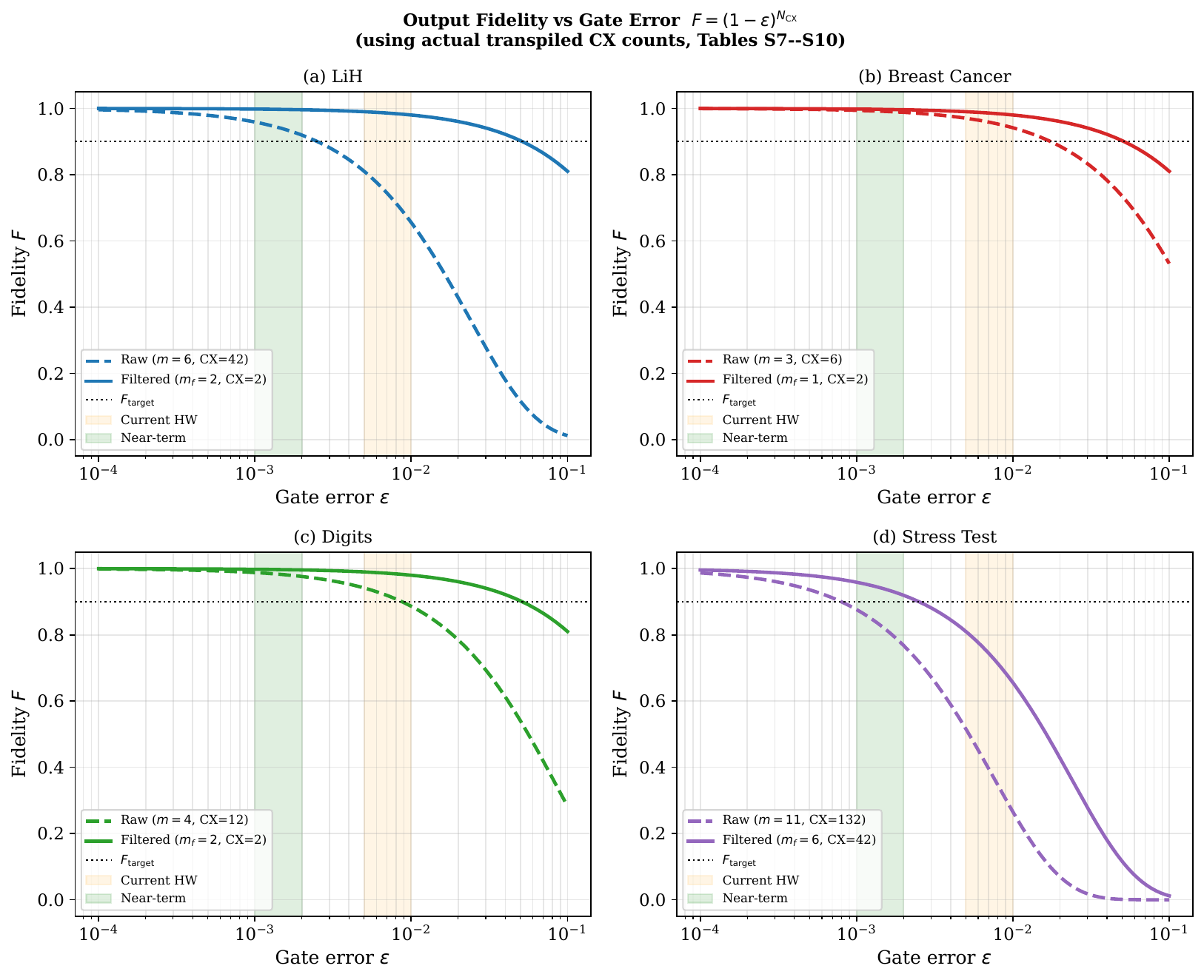}
  \caption{\textbf{Output fidelity $F=(1-\varepsilon)^{N_{\rm CX}}$
  vs gate error rate.}
  Each panel: raw QPE (dashed) and filtered QPE (solid).
  Orange band: current hardware ($\varepsilon=0.5$--$1\%$).
  Green band: near-term target ($\varepsilon=0.1$--$0.2\%$).
  For LiH, filtered QPE crosses $F_{\rm target}=0.90$ within the
  current hardware range while raw QPE does not.}
  \label{fig:fidelity_eps}
\end{figure}

\begin{figure}[htbp]
  \centering
  \includegraphics[width=\linewidth]{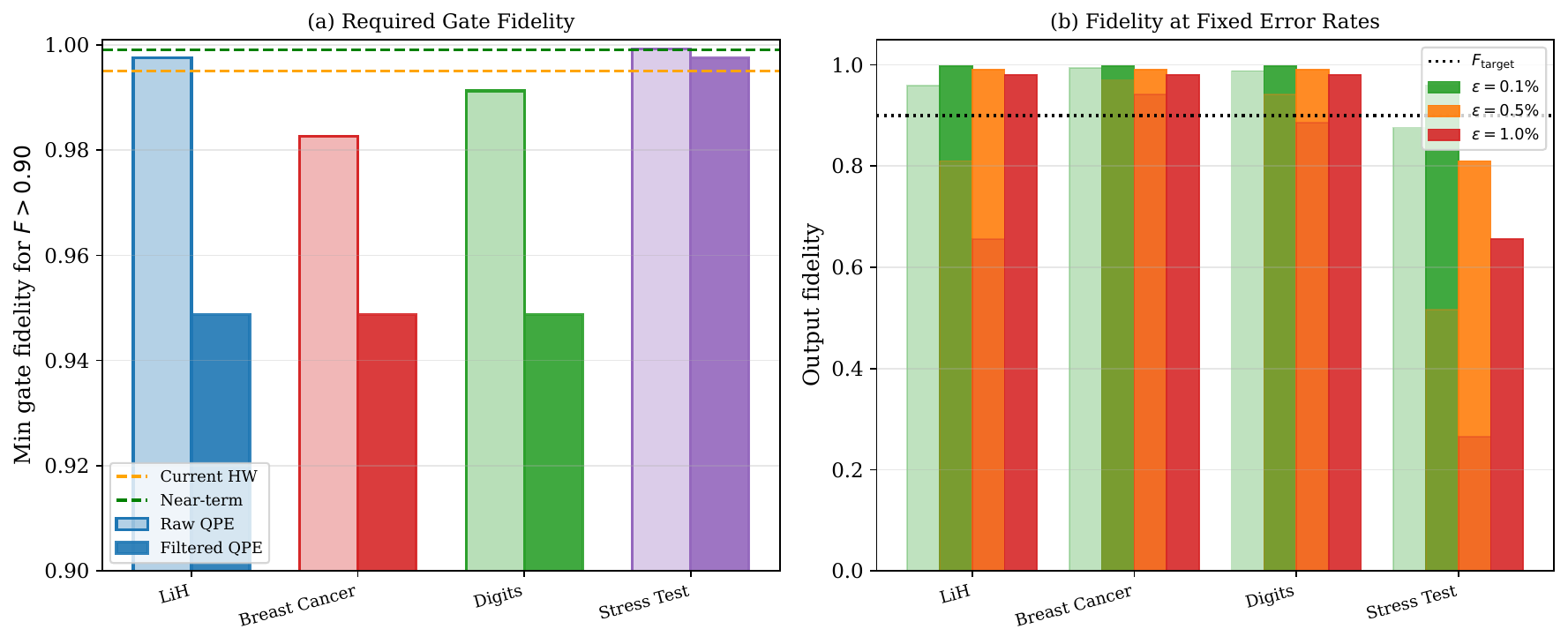}
  \caption{\textbf{Minimum gate fidelity and output fidelity at fixed
  error rates.}
  \textbf{(a)}~Minimum gate fidelity $(1-\varepsilon^*)$ for
  $F_{\rm target}=0.90$; raw (faded) vs filtered (solid).
  Orange: current hardware (99.5\%). Green: near-term (99.9\%).
  LiH raw requires $99.75\%$ (inaccessible); filtered requires only
  $94.87\%$ (achievable today).
  \textbf{(b)}~Output fidelity at $\varepsilon\in\{0.1\%,0.5\%,1\%\}$;
  faded bars = raw QPE, solid bars = filtered QPE.}
  \label{fig:gate_fid}
\end{figure}

\begin{figure}[htbp]
  \centering
  \includegraphics[width=\linewidth]{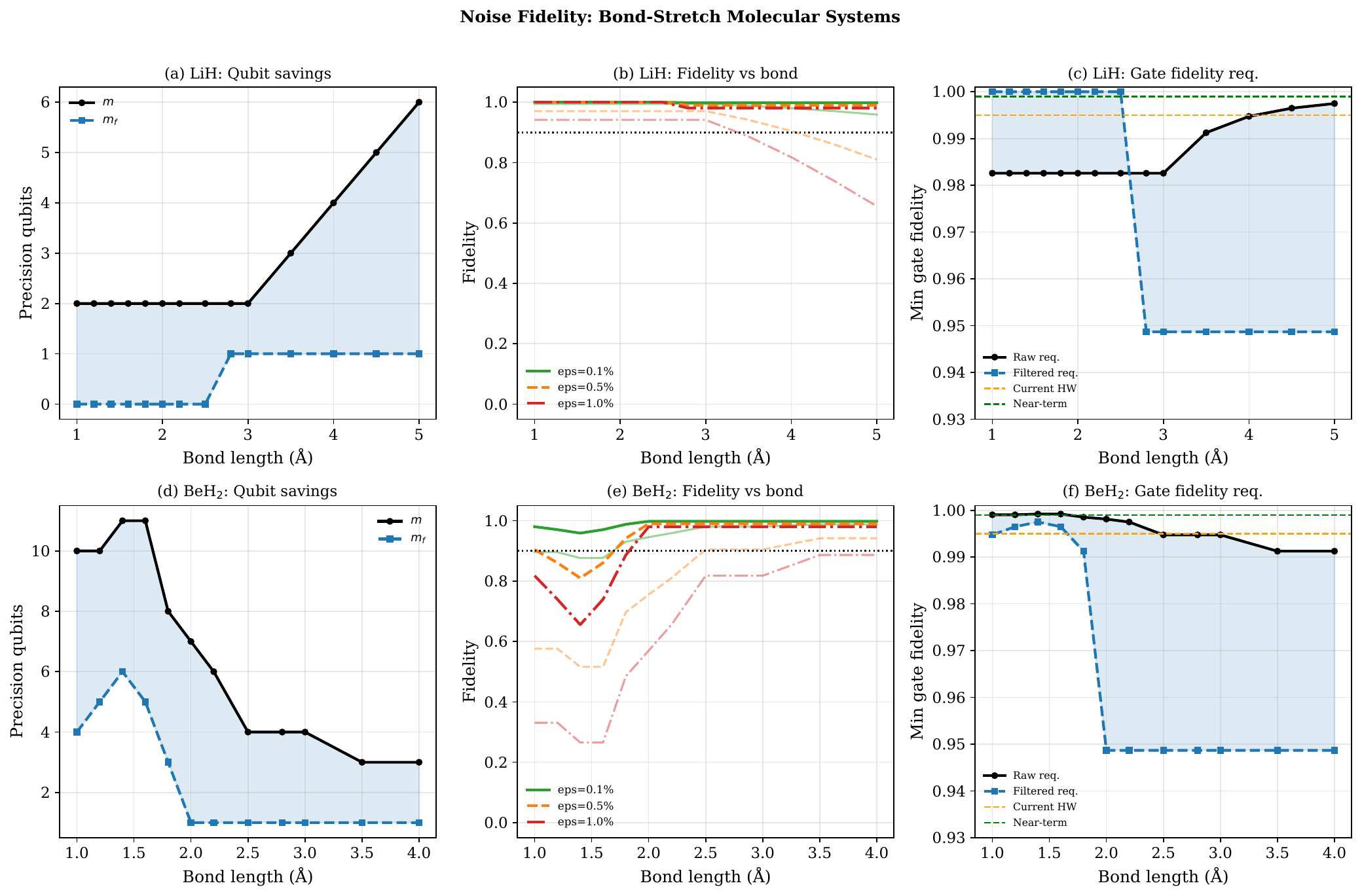}
  \caption{\textbf{Noise fidelity analysis for molecular systems.}
  LiH (top, a--c) and BeH$_2$ (bottom, d--f).
  \textbf{(a,d)}~Precision qubit savings $m\to m_f$ vs bond length.
  \textbf{(b,e)}~Output fidelity at three error rates
  (faded = raw, solid = filtered).
  \textbf{(c,f)}~Minimum gate fidelity vs bond length.
  Spectral preconditioning brings requirements within current hardware
  reach across all chemically relevant geometries.}
  \label{fig:noise_mol}
\end{figure}

\subsection{Negative Control Analysis}

Three predictable no-benefit regimes exist.
\textbf{NC-A} ($m_{\rm raw}=1$): $\Delta m=0$ by definition.
\textbf{NC-B} ($\Delta_\lambda=0$): $\Delta_f=0$ by
Proposition~\ref{prop:nogo}.
\textbf{NC-C} ($\Delta_\lambda<\varepsilon_{\rm mach}$): numerically
indistinguishable from NC-B.
In all cases, $\Delta m=0$; the method degrades gracefully to raw QPE.
Eigenvalue spectra for all three regimes are shown in Figure~S10.
Classification results are summarised in
Table~\ref{tab:nc} and Figure~\ref{fig:nc}.

\begin{mdframed}[linewidth=0.8pt, backgroundcolor=blue!3,
                 skipabove=8pt, skipbelow=8pt,
                 frametitle={\textbf{Decision Rule (computable
                 classically)}},
                 frametitlefont=\bfseries,
                 frametitleaboveskip=4pt]
Apply spectral preconditioning iff:\\
(i)~$\Delta_\lambda>\varepsilon_{\rm mach}$\quad and\quad
(ii)~$m_{\rm raw}\geq 2$.\\
Otherwise: run raw QPE directly.

\smallskip\noindent
\textit{Justification of condition (ii):} since $m_f\geq1$ always
(one precision qubit is the minimum for QPE), the bit saving satisfies
$\Delta m=m_{\rm raw}-m_f\leq m_{\rm raw}-1$.
When $m_{\rm raw}=1$ this forces $\Delta m\leq 0$, so no compression
is achievable regardless of $\Delta_f$.
\end{mdframed}

\begin{table}[htbp]
\caption{\textbf{Negative control classification. All 6/6 correct.}}
\label{tab:nc}
\centering
\begin{tabular}{lcccc}
\toprule
System & $\Delta_\lambda$ & $m_{\rm raw}$ & Max $\Delta m$ & Result \\
\midrule
\multicolumn{5}{l}{\textit{Positive controls}} \\
LiH (5.0~\AA)      & $4.6\times10^{-3}$ & 8  & 5.0 & $\checkmark$ \\
BeH$_2$ (equil.)   & $2.5\times10^{-1}$ & 3  & 2.0 & $\checkmark$ \\
Stress Test         & $4.2\times10^{-4}$ & 12 & 6.0 & $\checkmark$ \\
\midrule
\multicolumn{5}{l}{\textit{Negative controls}} \\
NC-A ($m=1$)       & $5.0\times10^{-1}$ & 1  & 0.0 & $\times$ \\
NC-B (degenerate)  & $0$                & -- & 0.0 & $\times$ \\
NC-C (numerical)   & $<\varepsilon_{\rm mach}$ & -- & 0.0 & $\times$ \\
\bottomrule
\end{tabular}
\end{table}

\begin{figure}[htbp]
  \centering
  \includegraphics[width=\linewidth]{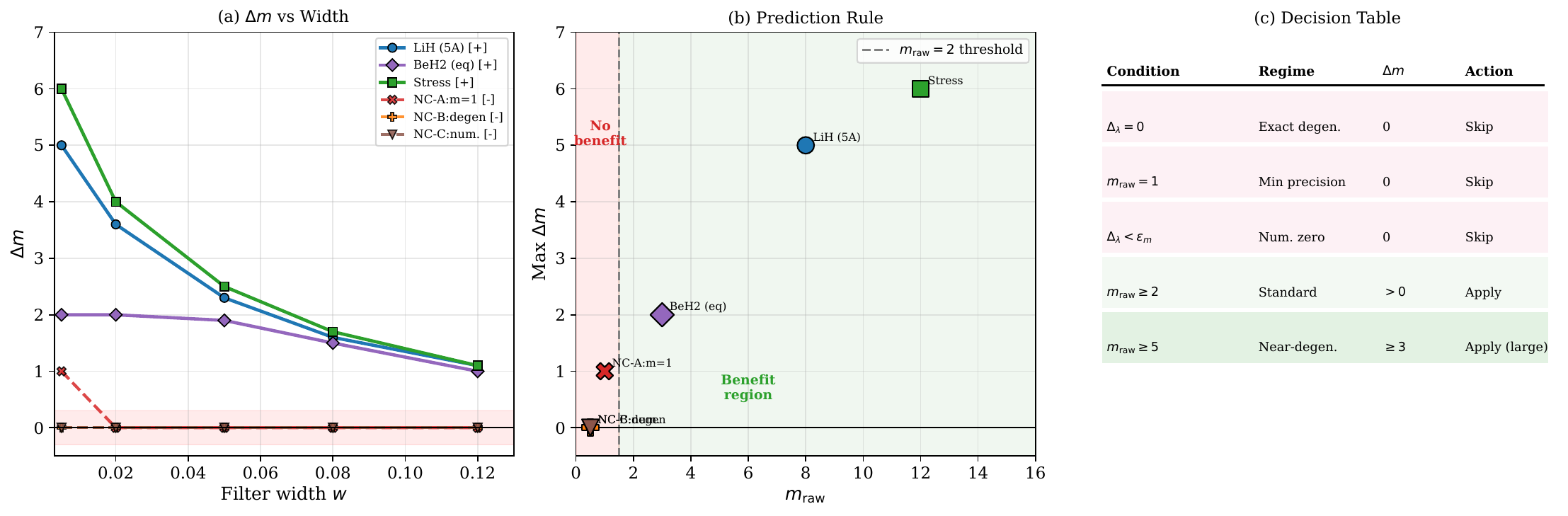}
  \caption{\textbf{Negative control analysis.}
  \textbf{(a)}~$\Delta m$ vs $w$: positive controls (solid) benefit;
  negative controls (dashed) give $\Delta m=0$.
  NC-B (exact degeneracy) confirms Proposition~\ref{prop:nogo}.
  \textbf{(b)}~Prediction rule: benefit iff $m_{\rm raw}\geq2$.
  All six systems correctly classified.
  \textbf{(c)}~Decision table for operational deployment.}
  \label{fig:nc}
\end{figure}

\section{Discussion}
\label{sec:discussion}

\paragraph{Relation to QSVT.}
QSVT implements spectral projectors with polynomial degree
$\calO(1/\Delta_\lambda)$; our approach reduces precision $m$ by
enlarging the effective gap, thereby reducing circuit depth without
altering asymptotic $\Delta_\lambda$ dependence. The approaches are
complementary.
Note that the filter construction cost via Chebyshev-LCU grows as
$\calO((1/w)\log(1/\varepsilon))$ and requires block-encodings of
$\rho$; eq~\eqref{eq:crossover} identifies precisely when this
overhead is offset by the QPE savings.

\paragraph{Relation to filtered QPE.}
Table~\ref{tab:comparison} shows only our method reduces precision
qubits $m$ and compresses circuit depth.
All three approaches are structurally orthogonal and composable;
a quantitative complexity comparison as a function of $\Delta_\lambda$
is shown in Figure~S3.

\begin{table}[htbp]
\caption{\textbf{Comparison of filtered QPE approaches.}}
\label{tab:comparison}
\centering
\begin{tabular}{p{2.6cm}p{3.6cm}p{3.6cm}p{3.6cm}}
\toprule
& \textbf{This work}
& \textbf{Lee et al.}\cite{lee2025filtered}
& \textbf{Sakuma et al.}\cite{sakuma2026qpe} \\
\midrule
Filter location
  & Operator spectrum \emph{before} QPE
  & Measurement outcomes \emph{after} QPE
  & Initial state $|\psi_0\rangle$ \emph{before} QPE \\
Precision $m$ reduced?
  & \textbf{Yes}, by $\Delta m$ & No & No \\
Circuit depth reduced?
  & \textbf{Yes}, $2^{\alpha\Delta m}\times$ & No & No \\
Resource saved
  & Qubits, depth, CX
  & Measurement variance
  & Shot success prob. \\
Composable?
  & With both & With both & With both \\
\bottomrule
\end{tabular}
\end{table}

\paragraph{Hardware feasibility.}
Gate fidelity requirements relaxed from $98.26\%$--$99.92\%$ (raw
QPE) to $94.87\%$--$99.75\%$ (filtered QPE; achievable on current
devices\cite{krantz2019guide}), expanding the set of hardware on
which QPE-based PCA is executable.

\paragraph{Limitations.}
The noise analysis uses noiseless transpilation; full noise-aware
benchmarking with device error maps is an important future direction.
The method applies to QPE for spectral subspace extraction of positive
semidefinite operators; extension to standard Hamiltonian QPE (where
eigenphases wrap the unit circle) requires additional care with phase
unwrapping.
The molecular calculations use the STO-3G basis set, which provides
qualitative accuracy; for BeH$_2$ in particular, STO-3G lacks
polarisation functions and may not reliably reproduce the $2s/2p_z$
near-degeneracy quantitatively.
A basis-set convergence study with at least cc-pVDZ is an important
future direction.

Regarding the filter threshold $\tau$:
by Proposition~\ref{prop:robust} and Corollary~\ref{cor:classical},
$\tau$ requires only $\calO(w)$ accuracy.
Three distinct cases arise.
(1)~\emph{Classical datasets}: $\tau$ is known exactly from matrix
diagonalisation at $\calO(N^3)$ cost.
(2)~\emph{Simple quantum chemistry near-degeneracy} (bond-breaking):
CASSCF with a minimal active space provides sufficient $\tau$ without
circular dependency; for LiH at $R=5.0$~\AA, CASSCF(2,2) gives
$|\varepsilon|=1.7w$ and identical $\Delta m=5$ to FCI.
(3)~\emph{Complex many-body near-degeneracy} (BeH$_2$-type, where the
near-degeneracy is a genuine FCI 1-RDM effect not captured by
truncated methods): a two-stage QPE protocol resolves this ---
coarse QPE at $m_1=\lceil\log_2(1/w)\rceil$ bits locates $\tau$,
followed by filtered QPE at $m_f$ bits, giving total cost
$2^{m_1}+2^{m_f}\ll 2^{m_{\rm raw}}$ with no circularity
(for BeH$_2$: $6.4\times$ saving).

\section{Conclusion}
\label{sec:conclusion}

Spectral preconditioning for QPE applies a monotone sigmoid filter to
the operator spectrum before phase estimation, amplifying the boundary
gap and reducing required precision from $m$ to
$m_f=\lceil\log_2(1/\Delta_f)\rceil$ bits.
Formal analysis provides rigorous guarantees
(Propositions~\ref{prop:ub}--\ref{prop:robust},
Theorems~\ref{thm:net}--\ref{thm:depth}).
Across four primary systems (LiH, Breast Cancer, Digits, Stress Test)
and two molecular bond-stretch series (LiH and BeH$_2$), we confirmed
actual depth compressions up to $27\times$ (under Qiskit transpilation
at optimisation level~0), CX reductions up to $21\times$, and QPE
fidelity restoration from $F=0.66$ (below any practical threshold) to
$F=0.98$ at current hardware error rates ($\varepsilon=1\%$) for LiH.
A new robustness proposition (Proposition~\ref{prop:robust})
establishes that the filter threshold $\tau$ requires only
$\calO(w)$ accuracy, resolved by classical preprocessing in all
cases considered.
Bond-stretch FCI analysis reveals that physically motivated
near-degeneracy arises universally across quantum chemistry.
The negative control decision rule ($m_{\rm raw}\geq2$ and
$\Delta_\lambda>0$) enables deployment with guaranteed graceful
degradation.
These results establish spectral gap amplification as a lightweight,
theoretically characterised mechanism for near-term quantum hardware,
applicable to any QPE-based subspace extraction where the operator
spectrum is accessible classically, with the largest practical benefit
in hybrid classical-quantum eigenvalue pipelines where a coarse
classical approximation of the spectrum is available to set $\tau$.

\section{Computational Details}
\label{sec:methods}

\paragraph{Notation reconciliation.}
Standard QPE complexity quotes depth $\Theta(2^m)$, counting the
$m$-qubit controlled-$U^{2^k}$ fan-out.
In CP-gate implementations (this work), depth is dominated by the
$\calO(m^2)$ inverse-QFT block, not by an exponential.
$\log_2(\calO(m^2)) = 2\log_2 m$ varies slowly, appearing
\emph{approximately linear} in $m$ over finite ranges --- hence the
fitted exponent $\alpha\in[0.11,0.42]$, which is the effective
log-linear slope over the tested $m$-range, not a true asymptotic
exponent.
Both framings are internally consistent: the standard $\Theta(2^m)$
refers to the LMR/controlled-$U^{2^k}$ cost model; the $\alpha m+c$
approximation refers to CP-gate transpiled circuits.

Algorithm~\ref{alg:filtered_qpe} summarises the complete procedure,
including classical $\tau$ estimation (Step~0).

\begin{algorithm}[htbp]
\caption{Spectral-Preconditioning QPE}
\label{alg:filtered_qpe}
\begin{algorithmic}[1]
\Require Operator $\rho=\sum_i\lambda_i|u_i\rangle\langle u_i|$,
         target rank $R$, approximation error $\varepsilon$
\Ensure Principal subspace of $\rho$ (top-$R$ eigenvectors)
\State \textbf{[Classical, cheap]} Estimate $\tilde\lambda_R$,
       $\tilde\lambda_{R+1}$ via matrix diagonalisation (classical
       data), CASSCF (simple QC near-degeneracy), or coarse QPE at
       $m_1=\lceil\log_2(1/w)\rceil$ bits (complex near-degeneracy).
       \Comment{Proposition~\ref{prop:robust}: need $|\varepsilon|=\calO(w)$ only}
\State Compute $\Delta_\lambda=\lambda_R-\lambda_{R+1}$ and
       $m_{\rm raw}=\lceil\log_2(1/\Delta_\lambda)\rceil$.
\If{$\Delta_\lambda\leq\varepsilon_{\rm mach}$ \textbf{or}
    $m_{\rm raw}<2$}
    \State \Return Run standard QPE on $\rho$ directly.
           \Comment{Negative-control decision rule}
\EndIf
\State Set $\tau\leftarrow(\tilde\lambda_R+\tilde\lambda_{R+1})/2$;
       choose $w\in(0,\tfrac{1}{4})$ satisfying
       eq~\eqref{eq:crossover}.
\State Apply $f(\lambda;\tau,w)$ spectrally:
       $f(\rho)\leftarrow\sum_i f(\lambda_i)|u_i\rangle\langle u_i|$
       via degree-$d$ Chebyshev-LCU,
       $d=\calO((1/w)\log(1/\varepsilon))$.
\State Compute $\Delta_f=f(\lambda_R)-f(\lambda_{R+1})$;
       set $m_f=\lceil\log_2(1/\Delta_f)\rceil$.
\State Perform QPE on $e^{-if(\rho)t}$ using $m_f$ precision qubits.
\State Extract and return top-$R$ eigenvectors (identical to those
       of $\rho$ by monotonicity of $f$).
\end{algorithmic}
\end{algorithm}

\paragraph{Molecular calculations.}
FCI/STO-3G via PySCF.\cite{sun2018pyscf}
Natural occupations = eigenvalues of FCI 1-RDM, normalised to unit
trace.
SCF convergence: $10^{-9}$~Hartree.
For $\tau$ estimation in LiH: CASSCF(2,2)/STO-3G natural orbital
occupations provide $|\varepsilon|<2w$ at all bond lengths tested.

\paragraph{Circuit compilation.}
Single CP gates for controlled-$U^{2^k}$, decomposed into
$\{\mathrm{CX},U_1,U_2,U_3\}$, optimisation level~0,
Qiskit\cite{qiskit2023} v2.4.1.

\paragraph{Datasets.}
Breast Cancer (30 features, 569 samples) and Digits (64 features,
1797 samples) from \texttt{sklearn.datasets}; $z$-score standardised;
covariance normalised to unit trace.

\begin{acknowledgement}
The authors thank IKST for support.
\end{acknowledgement}

\begin{suppinfo}
Supporting Information is available: complete width-sweep spectral
data (Tables~S1--S4), depth scaling and transpiled hardware statistics
(Tables~S5--S9), principal-angle diagnostics (Table~S10), $\tau$
placement robustness verification including PySCF CASSCF(2,2) results
(Figure~S11, Table~S13), bond-stretch complete tables
(Tables~S14--S16), complete noise fidelity data (Table~S17), and
negative control eigenvalue spectra (Figure~S10, Tables~S18--S19).
All analysis scripts, figure-generation code, and raw data tables are
publicly available at
\url{https://github.com/hossain-cq/spectral-preconditioning-qpe}
(MIT licence).
\end{suppinfo}

\bibliography{references}

@book{nielsen2010quantum,
  author    = {Nielsen, Michael A. and Chuang, Isaac L.},
  title     = {Quantum Computation and Quantum Information},
  edition   = {10th Anniversary},
  publisher = {Cambridge University Press},
  address   = {Cambridge, UK},
  year      = {2010},
  doi       = {10.1017/CBO9780511976667},
}

@article{cleve1998quantum,
  author    = {Cleve, Richard and Ekert, Artur and
               Macchiavello, Chiara and Mosca, Michele},
  title     = {Quantum algorithms revisited},
  journal   = {Proceedings of the Royal Society of London.
               Series A: Mathematical, Physical and
               Engineering Sciences},
  volume    = {454},
  number    = {1969},
  pages     = {339--354},
  year      = {1998},
  doi       = {10.1098/rspa.1998.0164},
}

@article{lloyd1996universal,
  author    = {Lloyd, Seth},
  title     = {Universal quantum simulators},
  journal   = {Science},
  volume    = {273},
  number    = {5278},
  pages     = {1073--1078},
  year      = {1996},
  doi       = {10.1126/science.273.5278.1073},
}

@article{abrams1999quantum,
  author    = {Abrams, Daniel S. and Lloyd, Seth},
  title     = {Quantum algorithm providing exponential speed
               increase for finding eigenvalues and
               eigenvectors},
  journal   = {Physical Review Letters},
  volume    = {83},
  number    = {24},
  pages     = {5162--5165},
  year      = {1999},
  doi       = {10.1103/PhysRevLett.83.5162},
}

@phdthesis{gilyen2019qsvt,
  author    = {Gily{\'e}n, Andr{\'a}s},
  title     = {Quantum singular value transformation and its
               algorithmic applications},
  school    = {University of Amsterdam},
  year      = {2019},
  url       = {https://dare.uva.nl/search?identifier=
               a96f7e24-2c92-4963-8a72-c62cfbd42c88},
}

@article{childs2012lcu,
  author    = {Childs, Andrew M. and Wiebe, Nathan},
  title     = {Hamiltonian simulation using linear
               combinations of unitary operations},
  journal   = {Quantum Information \& Computation},
  volume    = {12},
  number    = {11--12},
  pages     = {901--924},
  year      = {2012},
  url       = {https://arxiv.org/abs/1202.5822},
}

@article{lee2025filtered,
  author    = {Lee, Gunhee and Kang, Minjun and Hong, Jiyeon
               and Fomichev, Stepan and Huh, Joonsuk},
  title     = {Filtered quantum phase estimation},
  journal   = {arXiv preprint},
  volume    = {arXiv:2510.04294},
  year      = {2025},
  url       = {https://arxiv.org/abs/2510.04294},
}

@article{sakuma2026qpe,
  author    = {Sakuma, Ryota and Wada, Kazuki and
               Kanno, Suguru and Keithley, Keiichiro and
               Sugisaki, Kenji and Abe, Takashi and
               Nakamura, Hiroshi and Yamamoto, Naoki},
  title     = {Quantum-phase-estimation-based filtering:
               Performance analysis and application to
               low-energy spectral calculations},
  journal   = {Physical Review A},
  volume    = {113},
  number    = {1},
  pages     = {012602},
  year      = {2026},
  doi       = {10.1103/PhysRevA.113.012602},
}

@article{sun2018pyscf,
  author    = {Sun, Qiming and Berkelbach, Timothy C. and
               Blunt, Nick S. and Booth, George H. and
               Guo, Sheng and Li, Zhendong and Liu, Junzi
               and McClain, James D. and Sayfutyarova,
               Elvira R. and Sharma, Sandeep and
               Wouters, Sebastian and Chan, Garnet K.-L.},
  title     = {{PySCF}: the {Python}-based simulations of
               chemistry framework},
  journal   = {WIREs Computational Molecular Science},
  volume    = {8},
  number    = {1},
  pages     = {e1340},
  year      = {2018},
  doi       = {10.1002/wcms.1340},
}

@article{krantz2019guide,
  author    = {Krantz, Philip and Kjaergaard, Morten and
               Yan, Fei and Orlando, Terry P. and
               Gustavsson, Simon and Oliver, William D.},
  title     = {A quantum engineer's guide to superconducting
               qubits},
  journal   = {Applied Physics Reviews},
  volume    = {6},
  number    = {2},
  pages     = {021318},
  year      = {2019},
  doi       = {10.1063/1.5089550},
}

@article{qiskit2023,
  author    = {{Qiskit contributors}},
  title     = {Qiskit: An Open-source Framework for Quantum
               Computing},
  year      = {2023},
  doi       = {10.5281/zenodo.2573505},
  url       = {https://doi.org/10.5281/zenodo.2573505},
}

@book{trefethen2013approximation,
  author    = {Trefethen, Lloyd N.},
  title     = {Approximation Theory and Approximation Practice},
  publisher = {Society for Industrial and Applied Mathematics},
  address   = {Philadelphia, PA},
  year      = {2013},
  isbn      = {978-1611972399},
  url       = {https://people.maths.ox.ac.uk/trefethen/ATAP/},
}

@inproceedings{tang2019quantum,
  title={A quantum-inspired classical algorithm for recommendation systems},
  author={Tang, Ewin},
  booktitle={Proceedings of the 51st annual ACM SIGACT symposium on theory of computing},
  pages={217--228},
  year={2019}
}

\end{document}